\documentclass[11pt]{amsart}
\usepackage[T1]{fontenc}
\usepackage[ansinew]{inputenc}
\usepackage{amssymb,amsmath,amsthm,eucal,mathrsfs,array,graphicx,xcolor}
\usepackage[all]{xy}
\usepackage{bbm}
\usepackage[hmargin=3cm,vmargin=2.5cm,bindingoffset=0.5cm]{geometry}

\newtheorem{thm}{Theorem}[section]
\newtheorem{lem}[thm]{Lemma}

\newtheorem{prop}[thm]{Proposition}
\newtheorem{cor}[thm]{Corollary}

\theoremstyle{definition}

\newtheorem{defn}[thm]{Definition}

\newtheorem{conj}[thm]{Conjecture}

\newtheorem{prob}[thm]{Problem}

\def \d {\delta}

\def \Si {\Sigma}

\def \O {\Omega}
          
\def \C {\mathbb{C}}

\def \E {\mathbb{E}}

\def \P {\mathbb{P}} 
\def \Q {\mathbb{Q}} 
\def \R {\mathbb{R}} 
\def \S {\mathbb{S}}
\def \Z {\mathbb{Z}}
\def \ind {\mathbbm{1}}

\def \B {\mathcal{B}}
\def \Ca {\mathcal{C}} 
\def \Diff {\mathcal{D}}

\def \l {\ell}

\def \SLE {\mathrm{SLE}}

\title{Classifying conformally invariant loop measures}
\author{St\'ephane Benoist}

\begin{document}
\begin{abstract}
We formulate a classification conjecture for conformally invariant families of measures on simple loops that builds on a conjecture of Kontsevich and Suhov \cite{KontSuh}. The main example in this class of objects was constructed by Werner \cite{Wer_loops}. We present partial results towards the algebraic step of this classification.

Solving this conjecture would provide another argument explaining why planar statistical mechanics models with conformally invariant scaling limits naturally occur in a one-parameter family, together with the dynamical characterization of SLE via Schramm's central limit argument, and with the conformal field theory point of view and its central charge parameter.

\end{abstract}

\maketitle
\tableofcontents

\section{Motivation}

We are interested in describing collections of measures on sets of simple loops that are conformally invariant scaling limits of interfaces found in two-dimensional statistical mechanics models. Let us first give an example of such a loop measure, coming from the Ising magnetization model.

\subsection{The Ising loop measure}

The Ising model on a subgraph $\mathcal{G}$ of the square grid $\Z^2$ at inverse temperature $\beta>0$
is a measure on configurations $\left(\sigma_{x}\right)_{x\in\mathcal{G}}$
of $\pm1$ spins located at the vertices of $\mathcal{G}$. A configuration appears with probability proportional to $\exp\left(-\beta H\left(\sigma\right)\right)$,
where the energy $H$ is given by $-H\left(\sigma\right)=\sum_{x\sim y}\sigma_{x}\sigma_{y}$
(the sum is over all pair of adjacent vertices of $\mathcal{G}$).

When the temperature is high (i.e. $\beta$ is small), we tend to see configurations that are very disordered at microscopic scale (i.e. spins are virtually independent): one can imagine that the heat agitation of each atom is enough to overcome the energy constraint, i.e. constraints due to the interactions between atoms. On the other hand, at low temperatures (i.e. high values of $\beta$), the bias $e^{-2\beta}$ will exclude configurations with too many disagreeing neighbors. The picture tends to be frozen (all spins have same sign) at the microscopic level. There is a unique critical parameter $\beta_c=\frac{1}{2}\ln\left(\sqrt{2}+1\right)$, where the Ising model exhibits an intermediary behavior between disordered and being frozen.

Interfaces of the critical Ising model are known to converge to a conformally invariant scaling limit \cite{SmiChe_Ising}, and this allows us to construct a continuous loop measure from discrete Ising interfaces. Given a simply-connected domain $\O$ in the plane, we approximate it for each $\delta>0$ by a discrete domain $\O^\delta$ which is a collection of faces of a square lattice of mesh size $\delta$. We consider the critical Ising model on the graph $\O^\delta$, with $+$ boundary conditions, i.e. we fix the spins on the boundary of $\O^\delta$ to be $+1$. To a spin configuration $\sigma$, we can associate a random collection of curves $c(\sigma)$: the set of all interfaces i.e. the set of loops on the graph dual to $\O^\delta$ that wind between $+$ and $-$ spins. We call $m_{\O^\delta}$ the measure on collections of loops $c$ that are interfaces of the Ising spin model. In the scaling limit $\delta\to 0$, the measure $m_{\O^\delta}$ converges (loops are compared using the supremum norm up to reparametrization) towards a measure $m_\O$ called CLE$_3$ \cite{BeHo_CLEIs}. The collection of measures $(m_\O)_\O$ is then conformally invariant: given two simply-connected domains $\O$ and $\O'$ in the plane, and a conformal isomorphism $\phi:\O\rightarrow \O'$ between them, we have
$$
m_{\O'} = \phi_\ast\left(m_{\O}\right).
$$

The convergence of the whole collection $c$ of Ising loops implies the convergence of the measures $\mu_{\O^\delta}$ on single Ising interface loops $\ell$
$$
\mu_{\O^\delta}({\rm d} \ell)= \E^{\sigma}\left[\sum_{\ell_0\in c(\sigma)} \delta_{\ell_0}({\rm d} \ell)\right] = \int_{c\in\mathcal{C}} \sum_{\ell_0\in c} \delta_{\ell_0}({\rm d} \ell)  m_{\O^\delta}({\rm d} c)
$$
to a conformally invariant collection of measures $\mu_\Omega$ (of infinite mass) which describes the loops of a CLE$_3$:
$$
\mu_{\O}({\rm d} \ell)=\int_{c\in\mathcal{C}} \sum_{\ell_0\in c} \delta_{\ell_0}({\rm d} \ell)  m_{\O}({\rm d} c).
$$

\subsection{Loops and interactions}

One can wonder whether the Ising loop measure $\mu_{\O}({\rm d} \ell)$ is characterized by the macroscopic interactions of the statistical mechanics model. One way to make sense of this question is by keeping track of interactions by investigating how the position of the boundary of the domain boundary influences the shape of the loops.

At the discrete level, we can do the computation. Let $\O'$ be a subdomain of $\O$, and let us consider a loop $\ell \subset {\O'}^\delta$. We can compute the respective likelihood to see the loop $\ell$ as an Ising loop in ${\O'}^\delta$ and $\O^\d$. This Radon-Nikodym derivative can indeed be written as a ratio of Ising partition functions:
\begin{equation}\label{eq:Z}
\frac{{\rm d}\mu_{{\O'}^\delta}}{{\rm d}\mu_{\O^\delta}}(\ell)=\frac{Z_{{\O'}^\delta\setminus\ell}Z_{\O^\delta}}{Z_{{\O'}^\delta}Z_{\O^\delta\setminus\ell}},
\end{equation}
where, for a discrete domain $\mathcal{G}$, the partition function is given by
$$
Z_{\mathcal{G}}=\sum_{\sigma\in\{\pm 1\}^{\mathcal{G}}}e^{-\beta_c H(\sigma)}.
$$ 
The right hand side of (\ref{eq:Z}) may be tractable in the scaling limit and converge to a function that we represent as $\exp(f(\ell,{\O'},\O))$, for a certain function $f$ (see Section \ref{sec:step1} for a description of the function $f$). This step is well-understood for the uniform spanning tree model \cite{BeDu_SLE2}.

The continuous Ising loop measure $\mu_\O$ should then react to domain restriction (i.e. boundary deformation) in the following way:
\begin{eqnarray}\label{eq:restr}
\frac{{\rm d}\mu_{\O'}}{{\rm d}\mu_\O}(\ell)=\exp\left(f(\ell,\O',\O)\right)\ind_{\ell\subset\O'}.
\end{eqnarray}
Now, suppose that we are given a family of measures $(\widetilde{\mu}_\O)_\O$ whose behavior under restriction is also given by the Ising restriction formula (\ref{eq:restr}).
Is $\widetilde{\mu}_\O$ the Ising loop measure, i.e is it true that for any domain $\Omega$, $\widetilde{\mu}_\O=\mu_\O$?

The above discussion could (at least conjecturally) be repeated for any discrete model exhibiting conformal invariance: the scaling limits of single loops in such models should fall in the class of families of measures that satisfy (\ref{eq:restr}) for some function $f$. This leads to the following question.
\begin{prob}\label{prob}
Classify all families of measures on single loops that can a priori appear as scaling limits, i.e. classify conformally invariant families of measures on loops, together with their restriction property.
\end{prob}
Some aspects of this classification are closely related to a question of Malliavin \cite{Mal_diffcircle} on existence of loop measures, as well as to a conjecture of Kontsevich and Suhov \cite{KontSuh}.

This classification would conjecturally provide another argument explaining why planar statistical mechanics models with conformally invariant scaling limits naturally occur in a one-parameter family. Arguments with similar conclusions include the dynamical characterization of SLE via Schramm's central limit argument - and the related CLE classification \cite{ShWe_CLEcharac}, as well as the conformal field theory classification, which (loosely) extracts a real parameter (the central charge) out of the action of the conformal group on local observables of the model.

\subsection{Classification of loop measures}\label{sec:steps}

The classification question (Problem \ref{prob}) splits into three steps on which we elaborate in this section.

The first step is to classify the possible restriction formulas, i.e. to understand what restriction functions $f(\ell,\O',\O)$ can appear in the formula (\ref{eq:restr}). Indeed, the function $f$ need to satisfy some algebraic conditions in order to appear as such a Radon-Nikodym derivative. The second step of the classification would be to prove uniqueness of the loop measures, i.e. to prove that there is at most one collection of loop measures for each type of boundary interaction $f$. Thirdly and finally, one should construct all these measures.

\subsubsection{Restriction functions}\label{sec:step1}

The first step, the question of classifying possible restriction formulas is in part an algebraic question that can be rephrased as a cohomology computation on the space of loop-decorated Riemann surfaces.

We conjecture in this paper (Conjecture \ref{conj}) that the restriction functions, up to absolute continuity of the underlying measures, a priori reduce to a one-parameter family for algebraic reasons. This one-parameter family can be written as $f=c M$, where the quantity $M(\ell,\O',\O)$ is (up to a factor) as in \cite[Proposition 2.29]{BeDu_SLE2} and can be interpreted as the mass of Brownian loops in $\Omega$ that intersect both $\ell$ and $\O\setminus\O'$ (see \cite[Section 4]{LW}). Moreover (with the right choice of normalization factor for $M$), the central charge $c$ is related to the SLE parameter $\kappa$ by 
$$
c=\frac{(3\kappa-8)(6-\kappa)}{2\kappa}.
$$
Conjecturally, loop measures exist only when $c\leq 1$ for probabilistic reasons \cite{KontSuh}. This situation is reminiscent of the classification of restriction measures \cite{LSW3} , that are a priori classified by one positive real parameter $\alpha>0$, and later shown to only exist for $\alpha\geq 5/8$ for probabilistic reasons. Moreover, note that quantities similar to the Brownian loop mass $c M(\ell,\O',\O)$ appear when one studies how chordal SLE$_\kappa$ depends on the boundary of the domain \cite[Section 7.2]{LSW3}.

\subsubsection{Characterization}\label{sec:step2}

The second step of the classification program, the uniqueness of the loop measure having a fixed restriction property, is a conjecture of Kontsevich and Suhov \cite{KontSuh}. It has been proved by Werner for $\kappa = 8/3$ \cite{Wer_loops}, which corresponds to trivial interactions with the boundary, i.e. $f(\ell,\Omega',\Omega)=0$. The same result was achieved in \cite{ChaPic_Werloop} by considering the structure of infinitesimal deformations of domains.

\subsubsection{Construction}\label{sec:step3}
Loop measures were built by Werner \cite{Wer_loops} for $\kappa=8/3$ as boundaries of Brownian loops. Loop measures for $\kappa=2$ were constructed as a scaling limit of a discrete loop-erased walk \cite{KasKen_RandomCurves,BeDu_SLE2}.

For general values of the SLE parameter $0<\kappa\leq 4$ or equivalently general values of the central charge $c\leq 1$ (this is the simple curve regime, and conjecturally covers all simple loop measures), the loop measures are constructed in a work in preparation \cite{BeDu_SLEkloops,BeDu_SLE4loops} by finding them as flow lines of the Gaussian free field in the imaginary geometry coupling of Miller and Sheffield \cite{MilSheIG1}.

\subsection{Content of this paper}

We now focus on the first step of the classification (Problem \ref{prob}), i.e. we want to understand all possible functions $f$ than can appear in the formula (\ref{eq:restr}). In Section \ref{sec:resfun}, we setup this algebraic question as a cohomology problem. In Section \ref{sec:cohom}, we discuss a couple of results on the corresponding cohomology group (Propositions \ref{prop:non-trivial} and \ref{prop:obstruction}).

\section{The algebraic classification of restriction functions}\label{sec:resfun}

By Riemann surface, we mean a surface $\Sigma$ equipped with a complex structure, with finitely many handles, finitely many boundary components, and no punctures. For our purposes, there is no loss of generality by thinking of $\Sigma$ as an open subset of the complex plane $\mathbb{C}$ with finitely many holes (see Proposition \ref{prop:ess}). An embedding $\Sigma_1 \hookrightarrow \Sigma_2$ is a conformal injective map from the Riemann surface $\Sigma_1$ to the Riemann surface $\Sigma_2$. A simple loop $\ell$ is the image of the unit circle by an injective continuous map: the map is considered up to reparametrization, including rerooting and orientation switching. The topology on loops we will use is the topology of uniform distance up to reparametrization, and we work with the corresponding Borel $\sigma$-algebra.

\subsection{Setup}\label{sec:setup}
We now consider families of $\sigma$-finite measures $(\mu_\Si)_\Si$ indexed by Riemann surfaces, where $\mu_\Si$ is a measure on the set of simple loops $\l$ on $\Si$. Implicit in this formalism is that such a family of measures is conformally invariant.
\begin{defn}
A family $(\mu_\Si)_\Si$ is Malliavin-Kontsevich-Suhov (MKS) if it satisfies a restriction property as in (\ref{eq:restr}): if $\Si_1\subset\Si_2$ are two Riemann surfaces, then
\begin{eqnarray}\label{eq:restr2}
\frac{{\rm d}\mu_{\Si_1}}{{\rm d}\mu_{\Si_2}}(\ell)= e^{f_\mu(\l,\Si_1,\Si_2)} \mathbbm{1}_{\l\subset\Si_1},
\end{eqnarray}
where $f_\mu(\l,\Si_1,\Si_2)$ is a priori an arbitrary function that we call restriction function.
\end{defn}

Suppose that we have two MKS families of measures $\mu$ and $\nu$ that are in the same absolute continuity class: one can find a (conformally invariant, i.e. coordinate independent) function $g(\l,\Si)$ defined on pairs formed by a Riemann surface $\Sigma$ and a loop $\l\subset\Si$ such that $\nu=e^{-g}\mu$. Then, the restriction function $f_\nu$ associated to the measure $\nu$ can be expressed as:
\begin{eqnarray}\label{eq:cb}
f_\nu(\l,\Si_1,\Si_2)=f_\mu(\l,\Si_1,\Si_2) + g(\l,\Si_2) - g(\l,\Si_1).
\end{eqnarray}
Moreover, the inverse operation also makes sense: given an MKS family of measures $\mu$ and a conformally invariant function $g(\l,\Si)$, we can define another MKS family of measures $\nu$ by $\nu=e^{-g}\mu$. The restriction function $f_\nu$ is then given by (\ref{eq:cb}). An example where two families of measures $\mu$ and $\nu$ are related in this way is when these families describe loops arising from the same statistical mechanics model, but with different boundary conditions.

Understanding the set of restriction functions $f$ modulo the equivalence relation (\ref{eq:cb}) is a first step towards classifying the absolute continuity classes of MKS families of measures  (as discussed in Section \ref{sec:steps}).

\subsection{The cohomology of loops}\label{sec:defcoho}
We are thus led to consider the following problem.

We call a configuration either
\begin{itemize}
\item a pair $(\l,\Si)$ consisting of a Riemann surface $\Sigma$ and a loop $\ell\subset\Si$, or
\item a triple $(\l,\Si_1,\Si_2)$ consisting of two Riemann surfaces $\Si_1\subset\Si_2$ and a loop $\ell\subset\Si_1$.
\end{itemize}
Which of the two we consider will be clear from context at any given point.

We define the set $\mathcal{C}$ of cocycles as the set of real-valued functions $f(\l,\Si_1,\Si_2)$ on configurations such that:
\begin{itemize}
\item $f$ is an additive cocycle, i.e. given three Riemann surfaces $\Si_1\subset\Si_2\subset\Si_3$ and a loop $\l\subset\Si_1$, we have that 
\begin{eqnarray}\label{eq:cocycle}
f(\l,\Si_1,\Si_3)=f(\l,\Si_1,\Si_2)+f(\l,\Si_2,\Si_3).
\end{eqnarray}
\item $f$ is conformally invariant (i.e. coordinate-independent).
\item $f$ is continuous in $\l$ (for the topology of uniform convergence up to reparametrization).
\end{itemize}
Note that the first item is satisfied by restriction functions. The second item is trivially satisfied from the formalism, but we insist on the fact that $f$ should not depend on how coordinates are chosen on $\Sigma_2$, e.g. how $\Sigma_2$ is embedded in a larger Riemann surface (see the non-trivial Lemma \ref{lem:rep}). 
The third item is a convenient way to enforce the measurability and the (local) integrability of $e^f$ in (\ref{eq:restr2}). However, this is more than a technical condition (see the comment after Proposition \ref{prop:obstruction}).

Let us now define the set of coboundaries $\B$  as the set of real-valued functions $f(\l,\Si_1,\Si_2)$ on configurations such that there exists a function $g(\l,\Si)$ on configurations $\l\subset\Si$ that satisfies:
\begin{itemize}
\item $f(\l,\Si_1,\Si_2)=g(\l,\Si_2)-g(\l,\Si_1).$
\item $g$ is conformally invariant (i.e. coordinate-independent).
\item $g$ is continuous in $\l$ (for the topology of uniform convergence up to reparametrization).
\end{itemize}
Note that the function $g$ associated to a coboundary is unique up to a global additive constant, and thus we can always assume that $g(\S^1,\C\P^1)=0$. Moreover, note that every coboundary is a cocycle, i.e. $\B\subseteq\Ca$.

The classification of absolute continuity classes of MKS family of measures (as discussed in Section \ref{sec:setup}) amounts to understanding the cohomology of restriction functions, i.e. to understand the set of all cocycles modulo coboundaries. 
\begin{conj}\label{conj}
The cohomology group $\mathcal{H}=\Ca\slash \B$ is a one-dimensional real vector space.
\end{conj}
We use here the word cohomology in the sense of understanding the quotient of a space $\Ca$ with additive properties such as (\ref{eq:cocycle}) by telescopic sums $\B$. The cohomology space $\mathcal{H}$ carries information on the structure of the space of all loop-decorated Riemann surfaces modulo conformal equivalence.

\section{On the cohomology of loops}\label{sec:cohom}

\subsection{The cohomology is non-trivial}

\begin{prop}\label{prop:non-trivial}
The cohomology group $\mathcal{H}$ is non-trivial.
\end{prop}

We give a probabilistic proof that relies on the existence of SLE loop measures. It would be interesting to have a purely algebraic proof.

\begin{proof}
Consider the point in cohomology $f^{\SLE_2}(\ell,\Sigma_1,\Sigma_2)=-2M(\ell,\Sigma_1,\Sigma_2)$ that is associated to the $\SLE_2$ loop measure $\mu^{\SLE_2}$ built in \cite{BeDu_SLE2} (the precise definition of the mass of Brownian loop $M$ does not matter here). We argue that $f^{\SLE_2}$ cannot be a coboundary. Indeed, if it were, one could write $f^{\SLE_2}(\ell,\Sigma_1,\Sigma_2)=g(\ell,\Sigma_2)-g(\ell,\Sigma_1)$ for some function $g$. The loop measure $\nu=e^g\mu^{\SLE_2}$ would then satisfy the exact restriction property
\begin{equation}\label{eq:truerestr}
\nu_{\Si_1}=\nu_{\Si_2} \mathbbm{1}_{\l\subset\Si_1},
\end{equation}
for all Riemann surfaces $\Si_1\subset\Si_2$.

The only loop measure (up to global scaling) satisfying (\ref{eq:truerestr}) is the $\SLE_{8/3}$ loop measure $\mu^{\SLE_{8/3}}$ \cite{Wer_loops}, and so we would have $\nu=C \mu^{\SLE_{8/3}}$, i.e. $e^g\mu^{\SLE_2}=C \mu^{\SLE_{8/3}}$. However, the measure $\mu^{\SLE_{8/3}}$ does not belong to the absolute continuity class of $\mu^{\SLE_2}$ (e.g. because the Hausdorff dimension of an $\SLE_\kappa$ curve for $\kappa\leq 8$ is given by $1+\frac{\kappa}{8}$ \cite{Bef_Haus}), a contradiction.
\end{proof}

\subsection{The obstruction lies in the regularity of $g$}

We now prove that the obstruction to a cocycle being a coboundary lies in the regularity of $g$ (Proposition \ref{prop:obstruction}).

\begin{defn}
We call a configuration $(\ell,A)$ essential if $A$ is an annulus, and if the loop $\l$ is homotopically non-trivial in $A$, i.e. if $\ell$ disconnects the two boundary components of $A$. A configuration $(\l,A_1,A_2)$ is called essential if $(\l,A_1)$ and $(\l,A_2)$ are.
\end{defn}

We say that a loop $\l$ drawn on a surface $\Si$ is analytic, if we can find an annular neighborhood $A$ of $\l$ in $\Si$ and a conformal embedding $\phi:A\hookrightarrow\C\P^1$ such that the configuration $(\l,A)$ is essential and $\phi(\ell)=\S^1$. Note that this is equivalent to asking that there exists an analytic parametrization of the loop $\l$ by the unit circle $\S^1$.
\begin{defn}
We call a configuration $(\l,\Si)$ (resp. $(\l,\Si_1,\Si_2)$) analytic if the loop $\l$ is.
\end{defn}
Note that a configuration being analytic is not a condition on the roughness of the embedding $\partial\Si_1\hookrightarrow\Si_2$ (which may even be ill-defined).

We now prove, in the spirit of \cite{KontSuh}, that all the structure of restriction functions comes from essential configurations, i.e. from annular regions.
\begin{prop}\label{prop:ess}
If $f$ is a coboundary for essential configurations, then $f$ is a coboundary.
\end{prop}

\begin{proof}
By assumption, we can find a continuous function $g(\l,A)$ defined for configurations $\l \subset A$ where $\l$ is a homotopically non-trivial loop in an annulus $A$ and such that $f(\l,A_1,A_2)=g(\l,A_2)-g(\l,A_1)$ for all essential configurations $\l\subset A_1\subset A_2$. Given a configuration $\l\subset\Si$, let us pick an annulus $A\subset\Sigma$ such that $(\ell,A)$ is an essential configuration, and tentatively define $g(\l,\Si):=f(\l,A,\Si)+g(\l,A)$.
\begin{itemize}
\item The function $g$ does not depend on the choice of the annulus $A$ : for a configuration $\l\subset A' \subset A\subset \Sigma$, we have that
\begin{eqnarray}
f(\l,A',\Si)+g(\l,A')-f(\l,A,\Si)-g(\l,A)&=&g(\l,A')-g(\l,A)+f(\l,A',\Si)-f(\l,A,\Si)\nonumber\\
&=&-f(\l,A',A)+f(\l,A',A)=0.\nonumber
\end{eqnarray}
\item If for any annulus $A$ the function $g(\ell,A)$ is continuous, then the function $g(\ell,\Sigma)$ is continuous in $\l$ for any Riemann surface $\Sigma$.
\item $f$ and $g$ are related by the coboundary formula: given a configuration $\ell\subset \Si_1\subset\Si_2$, consider an annulus $A\subset\Sigma_1$ such that the configuration $(\ell, A)$ is essential. Then, we have that
\begin{eqnarray}
g(\ell,\Si_2)-g(\ell,\Si_1)&=&f(\l,A,\Si_2)+g(\l,A)-f(\l,A,\Si_1)-g(\l,A)\nonumber\\
&=&f(\ell,\Si_1,\Si_2).\nonumber
\end{eqnarray}
\end{itemize}
\end{proof}

\begin{lem}\label{lem:rep}
If $f$ is a cocycle, the function $f(\S^1,A,\C\P^1)$ (for configurations where $\S^1$ winds non-trivially around the annulus $A$) only depends on the conformal type of $(\S^1,A)$.
\end{lem}
In particular, the quantity $f(\S^1,A,\C\P^1)$ does not depend on the embedding $(\S^1,A)\hookrightarrow(\S^1,\C\P^1)$.

Before we give the proof of this Lemma, let us define the group $\Diff$ of analytic diffeomorphisms of the circle. An element $\psi\in\Diff$ is a map from the unit circle $\S^1$ to itself such that:
\begin{itemize}
\item The map $\psi$ is analytic: seeing $S^1$ as the quotient $\R/2\pi$, the map $\psi$ is a $2\pi$-periodic real-analytic map from $\R$ to itself.
\item The map $\psi$ is a bijection.
\item The derivative of $\psi$ does not vanish (together with the preceding items, this is equivalent to asking that $\psi$ admits an analytic inverse).
\end{itemize}
The group law on $\Diff$ is given by composition.
\begin{proof}
Given a cocycle $f$, we define a morphism $\rho$ from the group $\Diff$ of analytic diffeomorphisms of the circle to $(\R,+)$.

Pick a diffeomorphism $\psi\in\Diff$ and consider a small enough annular neighborhood $A$ of $\S^1$ such that $\psi$ extends to $A$ as an injective holomorphic map. We consider the map $\rho:\Diff \to \R$ given by $\rho(\psi) = f(\S^1,\psi(A),\C\P^1)-f(\S^1,A,\C\P^1)$.

\begin{itemize}

\item The quantity $\rho(\psi)$ does not depend on the choice of $A$ : given an annulus $A'$ such that $\S^1\subset A' \subset A$, we have that
\begin{eqnarray}
&&f(\S^1,\psi(A'),\C\P^1)-f(\S^1,A',\C\P^1)-f(\S^1,\psi(A),\C\P^1)+f(\S^1,A,\C\P^1)\nonumber\\
&=&f(\S^1,\psi(A'),\C\P^1)-f(\S^1,\psi(A),\C\P^1)+f(\S^1,A,\C\P^1)-f(\S^1,A',\C\P^1)\nonumber\\
&=&f(\psi(\S^1),\psi(A'),\psi(A))-f(\S^1,A',A)=0.\nonumber
\end{eqnarray}

\item The map $\rho$ is a group morphism. Indeed, let $\psi,\phi \in \Diff$, and let $A$ be an annulus such that $\psi_{|A}$ and $\phi_{|\psi(A)}$ are injective maps. Then, we have
\begin{eqnarray}
\rho(\phi\circ\psi)&=& f(\S^1,\phi\circ\psi(A),\C\P^1)-f(\S^1,A,\C\P^1)\nonumber\\
&=& f(\S^1,\phi\circ\psi(A),\C\P^1)-f(\S^1,\psi(A),\C\P^1)+f(\S^1,\psi(A),\C\P^1)-f(\S^1,A,\C\P^1)\nonumber\\
&=&\rho(\phi)+\rho(\psi).\nonumber
\end{eqnarray}
\end{itemize}

However, any morphism $\rho:\Diff\rightarrow(\R,+)$ needs to be trivial (Corollary \ref{cor:morph}, Appendix \ref{sec:app}). Hence, given two embeddings $(\S^1,A)$ and $(\S^1,\psi(A))$ of the same configuration in $\C\P^1$, $f(\S^1,\psi(A),\C\P^1)-f(\S^1,A,\C\P^1)=\rho(\psi)=0$.
\end{proof}

\begin{prop}\label{prop:obstruction}
If $f$ is a cocycle, there exists a (not necessarily continuous) function $g$ such that  $f(\l,A_1,A_2)=g(\l,A_2)-g(\l,A_1)$ on essential analytic configurations $\ell\subset A_1\subset A_2$.
\end{prop}
Analytic configurations being dense (and in light of Proposition {prop:ess}), only the (uniform) continuity of $g$ is missing to imply that any cocycle $f$ is a coboundary. However, this is not the case as the cohomology space $\mathcal{H}$ is non-trivial (Proposition \ref{prop:non-trivial}). The obstruction to any cocycle being a coboundary is hence a regularity constraint.

\begin{proof}
Given a cocycle $f$, let us build a function $g$ as claimed.

We look for such a function $g$ such that $g(\S^1,\C\P^1)=0$. We then want to define $g(\S^1,A):=-f(\S^1,A,\C\P^1)$ for all configurations $(\S^1,A)$ where $A$ is an annular neighborhood of the unit circle in the Riemann sphere. This is a coordinate-independent definition thanks to Lemma \ref{lem:rep}.

Given an analytic and essential configuration $\l\subset A$, let us cut a small enough annular neighborhood $A'$ of $\l$ in $A$ such that $(\l,A')$ is conformally equivalent (by a conformal isomorphism $\phi$) to a configuration $(\S^1,\phi(A'))$ where $\phi(A')$ is a subset of the Riemann sphere. We define $g(\l,A):=f(\l,A',A)+g(\S^1,\phi(A'))$.

\begin{itemize}
\item The function $g$ does not depend on the choice of $A'$ : take an annulus $A''$ such that $\l\subset A'' \subset A'$. Then
\begin{eqnarray}
&&f(\l,A',A)+g(\S^1,\phi(A'))-f(\l,A'',A)-g(\S^1,\phi(A''))\nonumber\\
&=&f(\l,A',A)-f(\l,A'',A)+g(\S^1,\phi(A'))-g(\S^1,\phi(A''))\nonumber\\
&=&-f(\l,A'',A')+f(\phi(\l),\phi(A''),\phi(A'))=0.\nonumber
\end{eqnarray}

\item The function $g$ does not depend on the choice of $\phi$, by Lemma \ref{lem:rep}.

\item $f$ and $g$ are related by the coboundary formula. Indeed, given an essential analytic configuration $\ell\subset A'\subset A$, let $A''$ be an annulus such that $\ell\subset A''\subset A'$. Then
\begin{eqnarray}
g(\ell,A)-g(\ell,A')&=&f(\ell,A'',A)-f(\S^1,\phi(A''),\C\P^1)-f(\ell,A'',A')+f(\S^1,\phi(A''),\C\P^1)\nonumber\\
&=&f(\ell,A',A).\nonumber
\end{eqnarray}

\end{itemize}
\end{proof}

\appendix 
\section{The group of analytic diffeomorphisms of the circle does not admit non-trivial morphisms to $(\R,+)$}\label{sec:app}

\begin{prop}\label{prop:diff-simple}
The group $\Diff^+$ of orientation-preserving analytic diffeomorphisms of the circle is perfect: any element of $\Diff^+$ can be written as a finite composition of commutators, i.e. elements of the form $f\circ g\circ f^{-1}\circ g^{-1}$.
\end{prop}

\begin{proof}
We proceed in several steps.
\begin{itemize}
\item The subgroup of the conformal transformations of the sphere $\C\P^1$ that fix the unit circle is isomorphic to PSL$_2(\R)$, and naturally embeds in the group $\Diff^+$ of orientation-preserving analytic diffeomorphisms of the circle, as the family of maps
$$
z\mapsto e^{i\theta}\frac{z+c}{\overline{c}z+1}.
$$
This subgroup of $\Diff^+$ contains all rotations $R_\theta$ of angle $\theta$, and is well-known to be perfect. In particular, all rotations in $\Diff^+$ are finite compositions of commutators.
\item For any element $f\in\Diff^+$ we define its rotation number $r(f)\in\R/\Z$ in the following way. Let us pick a lift $F:\R \rightarrow \R$, i.e. if $\pi:\R\to\R/\Z\simeq \S^1$ is the canonical projection, we pick a function $F$ such that $\pi\circ F=f\circ\pi$. The rotation number is then given by
$$
r(f)=\lim_{n\to\infty}\frac{F^{(n)}(1)}{n},
$$
where $F^{(n)}$ denotes the composition of $F$ with itself $n$ times. The rotation number (as a real number) comes with an ambiguity of $\Z$ resulting from the choice of a lift $F$.

Analytic diffeomorphisms $f$ whose rotation number belongs to a non-trivial subset $\Theta\subset\R/\Z$ (of full Lebesgue measure) are analytically conjugated to a rotation \cite{Herm_conjdiffcercle}: if $r(f)\in\Theta$, we can find an analytic diffeomorphism $h\in\Diff^+$ such that $f=h^{-1}\circ R_{r(f)}\circ h$.
\item For any diffeomorphism $f\in\Diff^+$, the map $\alpha\mapsto r(R_\alpha\circ f)$ is onto, as a periodic non-decreasing continuous map. For continuity, see e.g. \cite{Kuehn_rotation}: it follows from the fact that the rotation number $r(f)=\frac{p}{q}\in\Q$ if and only if the iterated map $f^{(q)}$ has a fixed point.
\end{itemize}

Hence, given an element $f\in\Diff^+$, we can find an angle $\alpha$ such that the rotation number $r(R_\alpha\circ f)=\theta\in \Theta$. This implies that there exists an element $h$ of $\Diff^+$ such that
$$
R_\alpha\circ f = h^{-1}\circ R_\theta \circ h.
$$

We can then express $f$ in the following way:
$$
f=R_{-\alpha}\circ \left(h^{-1}\circ R_\theta \circ h \circ R_\theta^{-1}\right)\circ R_\theta,
$$
which is a composition of two rotations and a commutator, hence a finite composition of commutators.
\end{proof}

\begin{cor}\label{cor:morph}
The group $\Diff$ of analytic diffeomorphisms of the circle does not admit non-trivial morphism to $(\R,+)$.
\end{cor}

\begin{proof}
Given a group morphism $\rho:G\rightarrow A$ taking values in an abelian group $A$, the kernel of $\rho$ is a group that contains all commutators of $G$. In particular, by Proposition \ref{prop:diff-simple}, given a group morphism $\rho:\Diff\rightarrow \R$, the subgroup $\Diff^+\subset\Diff$ of orientation-preserving diffeomorphisms is in the kernel of $\rho$.

Hence $\rho$ factors through $\rho:\Diff\to \Diff/\Diff^+\simeq\Z/2\Z \to \R$, where the first map is the canonical quotient map, and where the second map needs to be trivial, as there are no non-trivial morphisms from $\Z/2\Z$ to $(\R,+)$.
\end{proof}

\section*{Acknowledgements}
I would like to thanks Yves Benoist and Julien Dub\'edat for helpful discussions.

\bibliography{biblio}{}
\bibliographystyle{alpha}

\end{document}